\documentclass[default,iicol]{sn-jnl}

\jyear{2022}%

\theoremstyle{thmstyleone}%
\newtheorem{theorem}{Theorem}
\newtheorem{proposition}[theorem]{Proposition}%
\newtheorem{corollary}{Corollary}

\theoremstyle{thmstyletwo}%

\theoremstyle{thmstylethree}%

\begin{document}

\title[ ]{Cost free hyper-parameter selection/averaging for Bayesian inverse problems with vanilla and Rao-Blackwellized SMC Samplers}

\author*[1]{\fnm{Alessandro} \sur{Viani}}\email{viani@dima.unige.it}

\author[2]{\fnm{Adam M} \sur{Johansen}}\email{a.m.johansen@warwick.ac.uk}

\author[1]{\fnm{Alberto} \sur{Sorrentino}}\email{sorrentino@dima.unige.it}

\affil*[1]{\orgdiv{Dipartimento di Matematica}, \orgname{Universit\`a di Genova}, \orgaddress{\city{Genova}, \postcode{16146}, \country{IT}}}

\affil[2]{\orgdiv{Department of Statistics}, \orgname{University of Warwick}, \orgaddress{\city{Coventry}, \postcode{CV4 7AL}, \country{UK}}}

\abstract{In Bayesian inverse problems, one aims at characterizing the posterior distribution of a set of unknowns, given indirect measurements. For non-linear/non-Gaussian problems, analytic solutions are seldom available: Sequential Monte Carlo samplers offer a powerful tool for approximating complex posteriors, by constructing an auxiliary sequence of densities that smoothly reaches the posterior.

Often the posterior depends on a scalar hyper-parameter, for which limited prior information is available. In this work, we show that properly designed Sequential Monte Carlo (SMC) samplers naturally provide an approximation of the marginal likelihood associated with this hyper-parameter for free, i.e. at a negligible additional computational cost. The proposed method proceeds by constructing the auxiliary sequence of distributions in such a way that each of them can be interpreted as a posterior distribution corresponding to a different value of the hyper-parameter.
This can be exploited to perform selection of the hyper-parameter in \textit{Empirical Bayes} (EB) approaches, as well as averaging across values of the hyper-parameter according to some hyper-prior distribution in \textit{Fully Bayesian} (FB) approaches.

For FB approaches, the proposed method has the further benefit of allowing prior sensitivity analysis at a negligible computational cost. In addition, the proposed method exploits particles at all the (relevant) iterations, thus  alleviating one of the known limitations of SMC samplers, i.e. the fact that all samples at intermediate iterations are typically discarded.

We show numerical results for two distinct cases where the hyper-parameter affects only the likelihood: a toy example, where an SMC sampler is used to approximate the full posterior distribution; and a brain imaging example, where a Rao-Blackwellized SMC sampler is used to approximate the posterior distribution of a subset of parameters in a conditionally linear Gaussian model.}




\maketitle

\section{Introduction}
\label{Sec:Introduction}

In Bayesian inverse problems, one is interested in approximating the posterior distribution of a set of unobservable quantities, $x$, conditioned on indirect measurements, $y$ \cite{stuart2010inverse}. Often the posterior distribution depends on a scalar hyper-parameter, $\theta\in\Theta\subseteq\mathbb{R}$, e.g. the noise variance: one can either perform hyper-parameter selection with an Empirical Bayes (EB) approach targeting the conditional posterior $p^{\theta^\star}(x\mid y)$ with the hyper-parameter set to the value which maximizes the marginal likelihood, $\theta^\star:=\textrm{arg\,max}_{\theta\in\Theta}\{p^{\theta}(y)\}$, sometimes termed type-II maximum likelihood \cite{good1965estimation}, or else marginalize out the hyper-parameter through a Fully Bayesian (FB) approach, targeting the posterior $p(x\mid y)$. However, both approaches often result in costly procedures.

One relatively common tool for approximating posterior distributions arising in Bayesian inverse problems are Sequential Monte Carlo (SMC) samplers \cite{del2006sequential}. 
SMC samplers construct an artificial sequence of distributions such that the first one can be readily sampled from and the last one coincides with the distribution of interest; a set of particles is drawn from the first density, and evolves gradually to approximate each distribution in the sequence. 
 
In most implementations of SMC samplers for Bayesian inverse problems, the samples obtained at intermediate iterations are discarded, because intermediate iterations are only used to facilitate the approximation of the target distribution. Not directly using these samples, except perhaps to estimate a normalizing constant, seemingly results in a substantial waste of computational resources.
Indeed, we have recently witnessed a growing number of studies that attempt to exploit/recycle particles from previous iterations in the final estimates \cite{Gramacy2010,Drovandi2019,7339702}. Gramacy et al. \cite{Gramacy2010} propose to recycle particles at different iterations by considering a weighted sum of all the approximated distributions in order to maximise the Effective Sample Size (ESS). Alternatively, Nguyen et al. \cite{7339702} propose to combine particles from past SMC samplers iterations considering the so called \textit{Deterministic Mixture Weight estimator}; a solution derived to combine weighted particles drawn from different proposal distributions. Recently South et al. \cite{Drovandi2019} developed a method which allows the samples from each generation of the algorithm to be used to approximate integrals over a part of the state space.

In this work we show that, for a large class of hierarchical Bayesian inverse problems, the intermediate iterations of properly designed SMC samplers can be used to perform selection of the hyper-parameter and/or averaging with respect to it, making EB/FB approaches feasible. All of this has only a negligible additional computational cost and, in the case of averaging, it also entails recycling of the particles at intermediate iterations, thus reducing the typical waste of computational resources.

The key idea underlying the proposed method is to define the auxiliary sequence of distributions in such a way that each distribution is a posterior distribution conditioned on a different value for the hyper-parameter. Such construction turns out to be extremely simple under certain conditions, for instance when the hyper-parameter appears only in the likelihood and the likelihood belongs to the natural exponential family; under other circumstances, finding the right sequence can be more challenging. Given the sequence, the estimate of the normalizing constant, naturally produced by SMC samplers, corresponds to an estimate of the evidence for the specific value of the hyper-parameter, which then allows maximum likelihood or Bayesian inference on the hyper-parameter. 

We provide the right tempering sequence for two different models largely used in inverse problems:

\begin{itemize}
    \item when the likelihood belongs to the Natural Exponential Family (NEF): here the tempering sequence obtained by raising the likelihood to a growing power between zero and one results in a proper sequence of densities that can be interpreted as posterior distributions;
    \item when the conditional posterior for a subset of variables $x_1$ can be analytically computed, and an SMC sampler is used only to approximate the posterior on the remaining variables $x_2$. For this class of models, which includes among others Conditionally Linear Gaussian (CLG) models, the auxiliary distribution sequence devised for the first case does not, in general, work fine, therefore we devise alternative sequences that can be used fruitfully in two special subcases.
\end{itemize}

The most straightforward application of the proposed method is the context of additive Gaussian noise inverse problems; here the interest is in the estimation of the joint posterior distribution for the state variables and the noise variance or the posterior distribution for the state variables conditioned on the estimated value for the noise variance.

As a first examples we consider the problem of recovering the mean of a Gaussian distribution from noisy observations, showing that the proposed approach performs as well as alternative approaches but with significant advantages in computational time. Then we show numerical results for a real world problem encountered in source analysis of Magneto/Electro-Encephalography data, in this case we show that the proposed approach provides reliable results and a substantial reduction of computational cost with respect to alternative approaches.

\section{Motivating example: source estimation in magneto/ electro-encephalography}
\label{Sec:MotivatingExample}

Magneto-/Electro-EncephaloGraphy (M/E-EG) are two non-invasive medical imaging techniques that record the magnetic/electric field on the scalp; from these recordings, it is possible to estimate the underlying neural currents \cite{RevModPhys.65.413}. Using the \textit{dipolar} assumption, this problem consists of estimating an unknown number of point sources, called \textit{dipoles}, each one defined by two quantities:

\begin{itemize}
    \item a location in the brain volume, conveniently represented as the index $r$ of a cell of a discretized brain (or \textit{voxel}); dipole location is assumed to be fixed in time;
    \item a 3-D vector $q$ representing orientation and intensity of the neural current at the specified voxel, and changing dynamically in time.
\end{itemize}

The inference problem can be formalized as
\begin{subequations}
    \begin{gather}
        y (t) = \sum_{i=1}^d G(r_i)q_i(t)+\varepsilon(t) \label{Eq:MEG_Forward}\\
        \varepsilon(t) \sim \mathcal{N}\left(0, \theta^2 \Sigma\right) \label{Eq:MEG_Noise}
    \end{gather}
\end{subequations}
where: $t = 1,\dots, T$ is a time index; $y(t)$ is an array containing the data recorded by all M/E-EG sensors at time $t$; $d$ is the (unknown) number of dipoles; $G(r_i)$ is the so called lead-field matrix, representing the magnetic/electric field generated by a unitary dipole located at $r_i$; $\varepsilon(t)$ is additive Gaussian noise whose (spatial) covariance matrix $\Sigma$ is known up to a scale factor $\theta$.  

This model was originally adopted in \cite{sorrentino2013dynamic,Sorrentino_2014}, where all unknown parameters were sampled with an SMC sampler, leading to high computational cost for long time series; in \cite{sommariva2014sequential} a Rao-Blackwellized version was presented that imposed a Gaussian prior on the $q$ variables and exploited the CLG structure, allowing to treat long time series with reduced computational cost. Finally, in \cite{viani2021bayes} a hierarchical model was presented that overcomes the limitations of the Gaussian prior by using a hyper-prior on the prior variance, thus substantially reducing the dependence on this hyper-parameter. 
Defining $\mathbf{y}:=(y(1), \dots, y(T))$ and \\$\mathbf{q}_{1:d}:=(q_{1:d}(1), \dots, q_{1:d}(T))$, the posterior distribution decomposes as:
\begin{equation}
\begin{split}
        p^{\theta}&(d, r_{1:d}, \mathbf{q}_{1:d},\lambda \mid \mathbf{y}) =\\ &p^{\theta}(\mathbf{q}_{1:d} \mid \mathbf{y}, d, r_{1:d}, \lambda)p^{\theta}(d, r_{1:d},\lambda \mid \mathbf{y})
    \label{Eq:MEG_Posterior}
\end{split}
\end{equation}
where the conditional posterior $p^{\theta}(\mathbf{q}_{1:d} \mid \mathbf{y}, d, r_{1:d}, \lambda)$ can be computed analytically, and only the second factor on the right hand side of \eqref{Eq:MEG_Posterior} has to be approximated via Monte Carlo. Importantly, there remains a dependence on the hyper-parameter $\theta$, namely the overall noise level, whose value has to be estimated.

\section{SMC Samplers for Bayesian inverse problems}
\label{Sec:StateArt}
In this Section we provide a brief summary of a class of SMC samplers that are often used for the approximation of posterior distributions in Bayesian inference problems. Notice that SMC samplers can be applied in more general situations, not analyzed in this paper; for further details on general SMC samplers algorithms the reader is referred to \cite{del2006sequential,del2007sequential}.

Consider a Bayesian inference problem where the aim is to approximate the posterior distribution
\begin{equation}
    p(x \mid y) = \frac{p(x) p(y \mid x)}{p(y)}
    \label{Eq:BayesTheorem}
\end{equation}
where $y$ represents the data and $x$ the unknown parameters. The posterior distribution is often a complex distribution in a possibly high-dimensional space and is typically difficult to sample from directly.
SMC samplers provide an effective way to sample such complex distributions, and can be briefly summarized as follows. 

The fist step is to define a sequence of intermediate densities:
\begin{subequations}
    \begin{gather}
        \bigl\{p_t(x \mid y)\bigl\}_{t=0}^T,
        \label{Eq:SequenceDensity}\\
       p_T(x \mid y) = p(x \mid y),
       \label{Eq:condition_final}\\
       p_t(x \mid y) \simeq p_{t+1}(x \mid y),
       \label{Eq:SequenceSmoothness}
    \end{gather}
\end{subequations}
that ``smoothly'' transition from an easy-to-sample initial density $p_0$ to the posterior density $p_T$. Condition \eqref{Eq:SequenceSmoothness} is required in order to guarantee a smooth transition toward the target density and hence to allow a good approximation of $p_{t+1}$ to be obtained from the corresponding approximation of $p_t$.

A natural, but not mandatory, choice in Bayesian inference is to reach the posterior density by starting from the prior and increasing the power of the likelihood using the so called geometric bridge, or tempering path \cite{syed2021parallel,chopin2020introduction,bernton2019schrodinger,neal2001annealed}: 
\begin{subequations}
    \begin{gather}
        p_t(x \mid y) \propto p(x)p(y \mid x)^{\alpha_t}, \label{Eq:NaturalBayesianSequence}\\
        0 = \alpha_0 < \alpha_1 < ... < \alpha_T = 1. \label{Eq:SequenceExponents}
    \end{gather}
\end{subequations}

Once the sequence of distributions has been selected, SMC samplers work as follows:

    \begin{itemize}
        \item sample a set of $N$ weighted particles $\{\mathbf{x}^{(0)}; \mathbf{W}^{(0)}\}$ from the initial distribution $p_0$
        \item for $t=1,\ldots,T$:
        \begin{enumerate}
            \item perform one, or more, Markov Chain Monte Carlo (MCMC) step/s; such as Metropolis Hastings step/s
            
            \item perform an Importance Sampling (IS) step from the current distribution $p_{t-1}$ to the next distribution $p_t$ updating the un-normalized importance weights and normalizing them using the relations
                \begin{subequations}
                    \begin{align}
                            \mathbf{w}^{(t)} =& \mathbf{w}^{(t-1)} \frac{p_{t}(\mathbf{x}^{(t-1)}\mid y)}{p_{t-1}(\mathbf{x}^{(t-1)}\mid y)}, \label{Eq:ImportanceWeights}\\
                           \mathbf{W}^{(t)} :=& \frac{\mathbf{w}^{(t)}}{\sum_{n=1}^N \mathbf{w}^{(t)}_{n}}. \label{Eq:ImportanceWeightsNormalized}
                        \end{align}
                \end{subequations}
                
            At this point one obtains an approximation of the $t$-th distribution of the sequence as:
            \begin{equation}
                \hat p_t(x\mid y) = \sum_{n=1}^N \mathbf{W}_n^{(t)} \delta_{x_n^{(t)}}(x).
            \end{equation}
            In this step one also obtains an estimator of the normalizing constant of the distribution $p_t$, crucial for model selection in general and for the proposed method in particular. It can be easily evaluated; for simplicity, assuming that resampling occurs at every step, as the product over time of the average of the un-normalized importance weights at eah time:
            \begin{equation}
                \hat{p}_t(y) = \prod_{s=0}^t\frac{1}{N}\sum_{n=1}^{N}\mathbf{w}^{(s)}_{n}.
                \label{Eq:NormalizingConstant}
            \end{equation}
the expression in the case that resampling is conducted adaptively can be found,  is the corresponding product over resampling times of the average of the weights accumulated since the last resampling time (see, e.g., \cite[p. 1641]{guarniero2017iterated} for an explicit expression).

            \item perform a resampling step to avoid degeneracy of the importance weights \cite{douc2005comparison,gerber2019negative}. A widely used strategy is to perform resampling whenever the Effective Sample Size (ESS) (see, e.g, \cite{liu2008monte}) is under a fixed threshold.
        \end{enumerate}
    \end{itemize}

One important property of SMC samplers comes from equation \eqref{Eq:ImportanceWeightsNormalized} which allows the evaluation of the importance weights at time $t$ using only the particles at the previous step. This allows the order of the first two steps of the algorithm to be reversed, which further allows an adaptive choice of the actual sequence of densities, as defined in \eqref{Eq:SequenceDensity}, through an online selection of the next exponent  \cite{del2012adaptive,Sorrentino_2014}.

\section{Selection/averaging of the hyper-parameter} 
\label{Sec:Selection/Averaging}

Let $\Theta\subseteq\mathbb{R}$ and consider a Bayesian inverse problem depending on a hyper-parameter $\theta\in\Theta$. We are now going to show how an SMC sampler can be used both to select a specific value for the hyper-parameter and/or to approximate the joint posterior distribution $p(x,\theta\mid y)$
at no additional cost with respect to the SMC sampler that approximates the conditional posterior $p^{\theta}(x\mid y)$.

The key idea underlying the proposed method is to construct an SMC sampler whose target distribution is  $p^{\theta^\star}(x\mid y)$ for some value $\theta^\star\in\Theta$, and whose intermediate distributions are posterior distributions corresponding to different values of the hyper-parameter for a set of values $\Theta_{0:T}:=\{\theta\in\Theta \;:\; \theta = \theta(t) ;\; t=0, \cdots, T\}$

\begin{equation}
\begin{split}
p_t^{\theta^*}(x\mid y) =& p^{\theta(t)}(x\mid y) \\= &\frac{p^{\theta(t)}(y\mid x)p^{\theta(t)}(x)}{p^{\theta(t)}(y)}.  
\end{split}
 \label{Eq:InterpretableSequence}
\end{equation}

Given the sequence above, one can  estimate pointwise the evidence for the hyper-parameter $p^{\theta}(y)$ for $\theta \in \Theta_{0:T}$ through Importance Sampling \eqref{Eq:NormalizingConstant}.  Under regularity assumptions for $p^{\theta}(y)$ w.r.t. $\theta$ one can interpolate this finite set of values to obtain a smooth approximation of the evidence and, assuming the availability of a hyper-prior $p(\theta)$, that we assume to be negligible outside a compact set $[\theta_{\textrm{min}}, \theta_{\textrm{max}}]$, an approximation of the marginal posterior $\hat{p}(\theta\mid y)$.

For an EB approach, one can first find the mode of the interpolating function properly weighted
\begin{equation}
    \bar{\theta}=\textrm{arg\,max}_{\theta\in [\theta_{\textrm{min}}, \theta_{\textrm{max}}]}\{\hat{p}(\theta\mid y)\},
\end{equation} 
where we assume that the range of $\Theta_{0:T}$ contains  $\theta^*,\theta_{\textrm{min}}$ and $\theta_{\textrm{max}}$. This can be done numerically by binary search, using importance sampling to estimate the marginal likelihood of values of $\theta$ between those in $\Theta_{0:T}$. We can then apply importance sampling to obtain an approximation of $p^{\bar{\theta}}(x\mid y)$. 

In order to avoid degeneration of importance weights, one should do importance sampling from $p^{\theta(\bar{t})}(x\mid y)$, where $\theta (\bar t)$ is the closest value to $\bar \theta$ such that the support and tails of $p^{\theta (\bar t)}(x\mid y)$ are larger and heavier, respectively, than those of  $p^{\bar \theta} (x\mid y)$; for instance, assuming that $\{\theta(t)\}_{t=0,\dots,T}$ is a decreasing sequence, and that the distributions tails become lighter as $\theta$ becomes smaller, we shall select the iteration 
\begin{equation}
    \bar t = \arg \min \{t \;:\; \theta(t) > \bar \theta\}.
    \label{Eq:tBar}
\end{equation}

For a FB approach one obtains an approximation of the posterior
\begin{equation}
    p(\theta\mid y) \propto p^{\theta}(y) p(\theta) \label{Eq:NoiseFactorization}
\end{equation}
for $\theta \in \Theta_{0:T}$, allowing to compute estimates such as the posterior mean or the maximum a posteriori for the hyper-parameter. 

In addition, it is possible to approximate the marginal posterior of the parameters
\begin{equation}
    \begin{split}
        p(x\mid y) &= \int p(x,\theta\mid y)d\theta \\&= \int p^{\theta}(x\mid y)p(\theta\mid y)d\theta \\&\propto \int p^{\theta}(x\mid y)p^{\theta}(y)p(\theta)d\theta .
    \label{Eq:Averaging}
    \end{split}
\end{equation}
that takes into account uncertainty on parameters deriving from uncertainty on the hyper-parameter.
This can be done by considering all particles at all iterations and re-weighting them
\begin{equation}
\begin{split}
\small
    \hat p(x\mid y)  =  \sum_{t=0}^T\sum_{n=1}^N \mathbf{W}^{(t)}_n &\delta_{x}(x^{(t)}_n) \hat p^{\theta(t)}(y) \\&p(\theta(t)) g^{(t)}(\Theta_{0:T})
    \label{Eq:NumericalIntegration}
    \end{split}
\end{equation}
where $g^{(t)}$ is a function representing the interpolation weights.

For example, in the case of a standard quadrature method such as the trapezoidal rule we get
\begin{equation}
    g_t(\Theta_{0:T}) = \left\{
    \begin{array}{ll}
      \|\theta_1 - \theta_0\|/2 & t = 0\\
      \|\theta_{t+1} - \theta_{t-1}\|/2 & 1<t<T \\
      \|\theta_T - \theta_{T-1}\|/2 & t = T\\
    \end{array}
    \right.
\end{equation}
but of course more sophisticated options are available \cite{zhou2016toward}.

The additional computational cost required for calculating \eqref{Eq:NoiseFactorization} - \eqref{Eq:NumericalIntegration} is negligible compared to the one needed for the approximation of $p^{\theta^\star}(x\mid y)$ directly with an SMC sampler employing likelihood tempering.

Moreover, the proposed FB approach has the advantage of making usage of particles at all iterations, thus avoiding the usual waste of computational resources.\\

As a last point we remark that, in the FB case, it is possible to modify the hyper-prior without re-running the SMC sampler: this allows cheap \textit{prior sensitivity analysis}, an important aspect to consider in applied Bayesian analyses, at a very small computational cost.\\

The construction of sequence \eqref{Eq:InterpretableSequence} is not always straightforward. In the following, we consider an inverse problem whose likelihood belongs to the NEF and the prior does not depend on the hyper-parameter, deriving sequence \eqref{Eq:InterpretableSequence} for two distinct cases: 
\begin{enumerate}
    \item the case where SMC samplers are used to approximate the full posterior distribution;
    \item the case where  the conditional posterior for a subset of variables $x_1$ can be analytically computed, and a Rao-Blackwellized SMC sampler is used to approximate the posterior on the remaining variables $x_2$.
\end{enumerate}

\subsection{Case 1: vanilla SMC samplers for the full posterior distribution} 
\label{Sec:NEF_Selection/Averaging}

As the likelihood belongs to the Natural Exponential Family (NEF) with natural scalar hyper-parameter $\theta\in\Theta\subseteq\mathbb{R}$, it has the following density 
\begin{equation}
    p^{\theta}(y\mid x) = \exp(\theta T(y\mid x)-A_{\theta}) \label{Eq:NEF_Likelihood}
\end{equation}
where $T(y\mid x)$ is a sufficient statistic and $A_{\theta}$ represents the log-normalizing constant.\\

\textbf{Proposition 1}
\textit{Let $p^{\theta} \in NEF$ with sufficient statistic $T$ and canonical parameter $\theta$ s.t. $p^{\theta}(x) = \exp(\theta T(x)-A_{\theta})$
and $\alpha\neq0$, then:}
\begin{equation*}
    [p^{\theta}(x)]^{\alpha} = \exp(A_{\alpha\theta}-\alpha A_{\theta}) p^{\alpha\theta}(x)
\end{equation*}

By the previous proposition, whose trivial proof is provided in Appendix, it is straightforward to show that the sequence \eqref{Eq:NaturalBayesianSequence} naturally provides an evaluation of the joint posterior distribution $p(x,\theta\mid y)$ for the set of values $\Theta_{1:T}=\{\theta\in\Theta \;:\; \theta = \theta\alpha_t ;\; t=1, \cdots, T\}$.

As an example, in the case of an inverse problem with additive Gaussian noise of unknown variance, the distributions of the sequence are posterior distributions corresponding to a decreasing variance $\sigma(t) =  \sigma^\star/\sqrt{\alpha_t}$\, where $\sigma^\star$ represents the noise standard deviation at the very last iteration of the SMC samplers.

\subsection{Case 2: Rao-Blackwellized SMC samplers}
\label{Sec:Selection/AveragingRaoBlackwell}

We now consider the case where the unknown variable $x$ can be decomposed into a pair of components 
$x = (x_1,x_2)$, and: 

\begin{itemize}
    \item the prior on $x_1$ belongs to the NEF with respect to a hyper-parameter $\lambda$ 
\begin{equation}
    p(x_1\mid\lambda)=\exp(\lambda S(x_1)-A_{\lambda})
    \label{Eq:NEF_prior}
\end{equation}
    where $S(x_1)$ is a sufficient statistic and $A_{\lambda}$ is the log-normalization constant;
    \item  the conditional posterior $p^{\theta}(x_1\mid x_2,\lambda,y)$ can be computed analytically. 
\end{itemize}
Under these assumptions, in the natural decomposition of the joint posterior density 
\begin{equation}
    p^{\theta}(x_1,x_2,\lambda\mid y) = p^{\theta}(x_1\mid x_2,\lambda,y)p^{\theta}(x_2, \lambda \mid y),
    \label{Eq:DecompositionFullPosterior}
\end{equation}
only the second factor of the right hand side needs to be approximated by an SMC sampler, thus reducing the variance of the importance weights and improving the quality of the approximation.
This class of models is widely used and appreciated in applications; in particular, an SMC sampler targeting the marginal posterior $p^{\theta}(x_2, \lambda \mid y)$ typically leads to more accurate estimates than an SMC sampler targeting the full posterior and using the same computational resources \cite{murphy2001rao}.

As a consequence of the hypothesis that both the likelihood \eqref{Eq:NEF_Likelihood} and the prior on the Rao-Blackwellized variable \eqref{Eq:NEF_prior} belong to the NEF, the marginal likelihood turns out to be

\begin{equation}
        \begin{split}
            p^{\theta}(y\mid x_2,\lambda) & = \int p^{\theta}(y\mid x_1,x_2,\lambda)p^{\lambda}(x_1) dx_1 \\
            & = \int \exp\biggl(
            \biggl\langle\begin{bmatrix}
            \theta\\
            \lambda
            \end{bmatrix},
            \begin{bmatrix}
            T(y\mid x_1,x_2)\\
            S(x_1)
            \end{bmatrix}
            \biggl\rangle \\ &- (A_{\theta}+A_{\lambda})\biggl) dx_1.
     \end{split}
    \label{Eq:MaringalLike}         
\end{equation}

In most cases, the marginal likelihood in equation \eqref{Eq:MaringalLike} does not have a closed form solution; below we show two special cases in which it does.

\subsubsection*{Additive statistic for the Likelihood}

If the statistic $T(y\mid x_1,x_2)$ of the full likelihood \eqref{Eq:NEF_Likelihood} is the sum of two statistics $T(y\mid x_1)$ and $T(y\mid x_2)$, then the marginal likelihood also belongs to the NEF with respect to the same parameter

\begin{equation}
        \begin{split}
            p^{\theta}(y \mid x_2,\lambda) &= \int \exp\biggl(\theta\bigl (T(y\mid x_1)+ T(y\mid x_2)\bigl )\\& +\lambda S(x_1)-(A^{(1)}_{\theta}+A^{(2)}_{\theta}+A_{\lambda})\biggl) dx_1 \\&\propto
         \exp\left(\theta T(y\mid x_2) - A^{(2)}_{\theta}\right).
        \end{split}
        \label{Eq:StatisticSum}
\end{equation}

For this particular subclass of models, the natural sequence \eqref{Eq:NaturalBayesianSequence} is still valid, as the marginal likelihood is still in the NEF.

\subsubsection*{Conditionally Linear Gaussian Model}
If both the full likelihood and the prior on $x_1$ have normal distribution 

\begin{subequations}
    \begin{align}
        p^{\theta}(y\mid x_1,x_2) \sim& \mathcal{N}(\mu(x_2)x_1, \theta^2\Sigma) \label{Eq:CLG_full}\\
        p(x_1\mid\lambda) \sim& \mathcal{N}(\eta, \Gamma_{\lambda}). \label{Eq:CLG_linear}
    \end{align}
\end{subequations}

it is well known that both the marginal likelihood \eqref{Prop:marginal_like} and the conditional posterior \cite{sommariva2014sequential} are Gaussian with known mean and variance

\begin{subequations}
    \begin{gather}
    \begin{split}
            p^{\theta}&(y\mid x_2,\lambda) \sim \\ &\mathcal{N}\left(\mu(x_2)\eta,\mu(x_2)^t \Gamma_{\lambda} \mu(x_2) + \theta^2\Sigma\right) \label{Eq:MarginalLikeCLG}
    \end{split}\\            
            p^{\theta}(x_1 \mid x_2,\lambda,y) \sim \mathcal{N}\bigl(\bar\mu, \bar\Sigma\bigl);
        \label{Eq:CLG_MarginalPost}
    \end{gather}
\end{subequations}

where:
\begin{itemize}
    \item $\bar\mu := \Gamma_{\lambda} \mu(x_2)^t (\mu(x_2) \Gamma_{\lambda} \mu(x_2)^t + \theta^2\Sigma)^{-1}y$
    \item $\bar\Sigma := \Gamma_{\lambda} - \Gamma_{\lambda} \mu(x_2)^t (\mu(x_2) \Gamma_{\lambda} \mu(x_2)^t)^{-1} \mu(x_2) \Gamma_{\lambda}$
\end{itemize}

In this case, the marginal likelihood is not in the NEF with respect to the parameter $\theta$ and the natural sequence \eqref{Eq:NaturalBayesianSequence} does not work. Indeed, by applying to the CLG model the same sequence constructed in the general case, one would get
    
\begin{subequations}
    \begin{gather}
        p_t^{\theta}(x_2,\lambda\mid y) \propto p(x_2,\lambda) p^{\theta}(y\mid x_2,\lambda)^{\alpha_t}
        \label{Eq:FirstNatrualChoice} \\
        \begin{split}
        p^{\theta}(y &\mid x_2,\lambda)^{\alpha_t} \propto \\ &\mathcal{N}\left(y;\; \eta\mu(x_2), \frac{1}{\alpha_t}\left(\mu(x_2)^t \Gamma_{\lambda} \mu(x_2) + \theta^2\Sigma\right) \right),
        \end{split}
        \label{Eq:FirstSequenceProblem}
    \end{gather}
\end{subequations}
    
since the marginal likelihood also embodies the prior on the marginalized variable $x_1$, the exponent also affects the prior for $x_1$; therefore, as already observed in \cite{sommariva2014sequential}, the distributions of this sequence cannot be considered as (marginals of) posterior distributions under the same prior.
    
Alternatively, one could consider  the sequence of marginals of the natural sequence for the approximation of the complete posterior density:

\begin{equation}
    \begin{split}
        p_t^{\theta}(x_2,\lambda\mid y)  & := \int p_t^{\theta}(x_1,x_2,\lambda\mid y) dx_1
        \label{Eq:SecondNatrualChoice}
    \end{split}
\end{equation}
    
However, also this choice leads to a sequence of distributions that cannot be interpreted as posterior distributions under different values of $\theta$; this happens because, as shown in Appendix (Corollary \ref{Cor:exponential_gaussian} and Proposition \ref{Prop:marginal_like}), the integral in \eqref{Eq:SecondNatrualChoice} is
\begin{equation}
\scriptstyle
    \begin{split}
        &p_t^{\theta}(x_2,\lambda\mid y) \propto \int p(x_1,x_2,\lambda) p^{\theta}(y\mid x_1,x_2,\lambda)^{\alpha_t} dx_1 = \\
        & p(x_2,\lambda)\int p(x_1\mid x_2,\lambda) p^{\theta}(y\mid x_1,x_2,\lambda)^{\alpha_t} dx_1 = p(x_2,\lambda) \\ & \ell_t(\lambda)\mathcal{N}\left(y;\; \eta\mu(x_2), \mu(x_2)^t \Gamma_{\lambda} \mu(x_2) + \frac{\theta^2}{\alpha_t} \Sigma\right),
        \label{Eq:SecondSequenceProblem}
    \end{split}
\end{equation}

where the Gaussian distribution can be interpreted as the marginal likelihood of the CLG model, with a different value of $\theta$, but the normalization constant
$\ell_t(\lambda)$, defined as in Proposition \ref{Cor:exponential_gaussian} in Appendix, depends on the hyper-parameter $\lambda$ and thus actually modifies the distribution.

However, it is not difficult to devise a proper sequence of intermediate distributions for the case of a CLG model. In fact, it is sufficient to explicitly remove the $\lambda$-dependent normalization factor from \eqref{Eq:SecondSequenceProblem} and construct the sequence as:

\begin{equation}
    \begin{split}
        p_t^{\theta}(x_2,\lambda \mid y) &\propto p^{\theta}(x_2,\lambda)\\ & \mathcal{N}\left(y;\; \eta\mu(x_2), \mu(x_2)^t \Gamma_{\lambda} \mu(x_2) + \frac{\theta^2}{\alpha_t}\Sigma\right). 
    \end{split}
\label{Eq:InnovativeSequence}
\end{equation}

With this definition we can apply the proposed approach to a CLG model while also exploiting Rao-Blakwellization.

\section{Toy Example}
\label{Sec:ToyExample}
We proceed with a numerical validation of the proposed approach by first using a toy example\footnote{Code available at: \url{https://github.com/alessandro-viani/ToyExample.git}}; following the arguments in Section \ref{Sec:Selection/Averaging}, we compare the results with natural alternatives for Fully Bayesian (FB) and Empirical Bayes (EB) approaches.

\subsection{Setup}
\label{Sec:Setup}
Consider an inverse problem where the aim is to reconstruct the mean of a Gaussian waveform of known variance $\sigma^2$, given noisy measurements $y(t)$, i.e.

\begin{subequations}
\begin{gather}
y(t) = \mathcal{N}(t;\; \mu, \sigma^2) + \varepsilon(t) \\
\varepsilon(t) \sim \mathcal{N}(0, \theta^2).
    \end{gather}
\end{subequations}
where $\mathcal{N}(t;\; \mu, \sigma^2)$ is the probability density function of a Gaussian of mean $\mu$ and standard deviation $\sigma$, evaluated at $t$.

We assume observations are available at $I$ points separated by unit intervals $\{ t_i\}_{i=1}^I$ and we want to make inference on the Gaussian mean.

\begin{figure}
    \centering
    \includegraphics[width=\columnwidth]{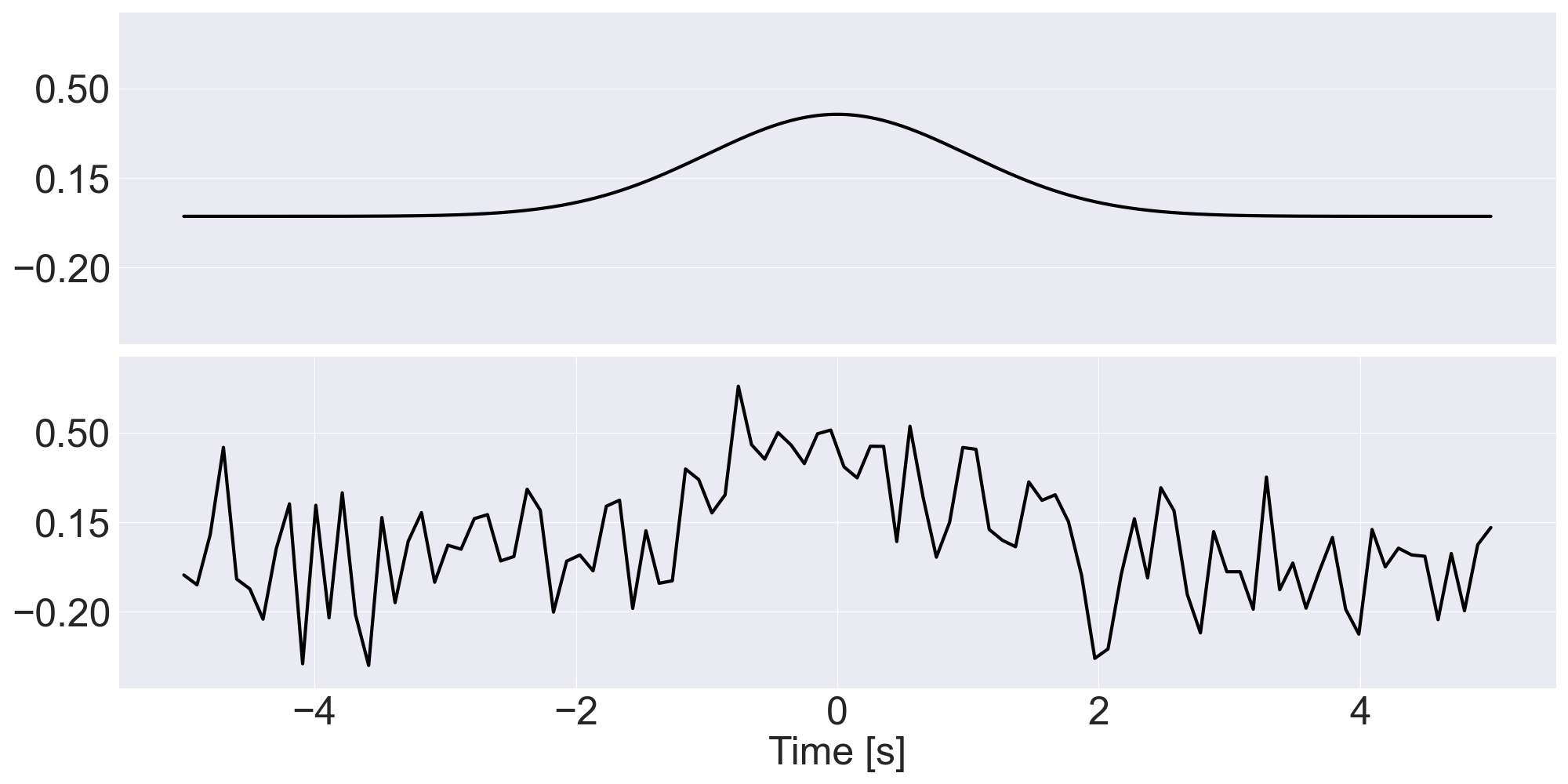}
    \caption{The figure shows in the first row an example of data without noise while in the second row the same data with the addition of noise.}
    \label{fig:DataToy}
\end{figure}
\subsection{Data Generation}

Data $\mathbf{y}=\left(y(t_1), \dots, y(t_I)\right)$ are generated considering $I=100$ measurements in the interval $[-5,\;5]$ obtained by perturbing the Gaussian density at each observation time independently with additive Gaussian noise of zero mean and standard deviation $\theta_{\textrm{true}} \sim \mathcal{U}[0.1,\; 0.2]$.

With these settings, we generate 100 independent realizations of the dataset in order to test the proposed algorithm.

\subsection{Prior and likelihood}
\begin{itemize}
    \item We assume $p(\mu) \sim \mathcal{U}\left([-5,\;5]\right)$ as a truncation of the Jeffrey's prior to the convex hull of the measurements;
    \item we assume $p(\theta) \sim \Gamma(2, 4\theta^\star)$, where $\theta^\star$ is an estimated value for the hyper-parameter;
    \item we assume conditional independence between observations given the parameter, obtaining a simple factorization for the likelihood
    \begin{equation}
        p^{\theta}(\mathbf{y}\mid \mu) = \prod_{t=1}^T p^{\theta}(y(t) \mid \mu).
        \label{Eq:prod_likelihood}
    \end{equation}

\end{itemize}

\subsection{Algorithm settings}
\label{Sec:AglorithmSettings}
For each of the 100 generated datasets, we compare the results obtained with the proposed method with the one obtained with a FB approach and an EB approach.

Each SMC sampler used has the following settings:
\begin{itemize}
    \item number of particles set to $100$ as a compromise between performances and quality of the approximation;
    \item $\theta^\star=\min\{\theta_{\textrm{true}}\}/2$; this allows the true value $\theta_{\textrm{true}}$ to be within the range of values explored by the Proposed method during SMC sampler iterations;
    \item number of iterations set to $500$, with the sequence of exponents growing exponentially in order to guarantee a smooth transit between intermediate distributions;
    \item resampling step performed by means of systematic resampling \cite{douc2005comparison} whenever the effective sample size is lower than half of the number of particles;
    \item Gaussian kernel for the MCMC step.
\end{itemize}

\subsection{Comparison with alternative approaches}

We compare the performances of the proposed method with those of two alternatives, one performing an Empirical Bayes approach and the other one performing a Fully Bayesian approach. In the following, particularly in the pictures, we denote by PropEB and PropFB the results obtained by the proposed method performing Empirical Bayes and Fully Bayesian approaches, respectively.

\subsubsection*{Empirical Bayes Approach}
For the EB approach we first obtain an estimate for the maximum a posteriori for the hyper-parameter:

\begin{equation}
    \hat\theta_{\textrm{MAP}} = \textrm{arg\,max}_{\theta}\{\hat{p}(\theta\mid \mathbf{y})\};
    \label{Eq:MAPtheta}
\end{equation}
where $\hat{p}(\theta\mid \mathbf{y}) $  is obtained by considering $M=100$ evenly spaced samples in the interval $[-5,\;5]$ :
\[\hat{p}(\theta\mid \mathbf{y})= \frac{1}{M}\sum_{i=1}^M p(\mu_i,\theta\mid \mathbf{y})\]
and then selecting the maximum value obtained over an evenly spaced grid of 500 points for $\theta \in [\theta^\star, \;50 \cdot \theta_{true}]$

Once an estimate for the hyper-parameter is obtained, we consider an SMC sampler targeting the posterior distribution $p^{\hat\theta_{\textrm{MAP}}}(\mu\mid \mathbf{y})$.

\subsubsection*{Fully Bayesian Approach}
For the FB approach we consider an SMC sampler targeting the posterior distribution $p(\mu,\theta\mid \mathbf{y})$, i.e. the hyper-parameter is sampled by the SMC sampler like all other parameters; the posterior distribution for the hyper-parameter is then obtained by marginalizing the joint distribution.

\subsection{Results}
\label{Sec:ReusltsToy}
\begin{figure*}
    \centering
    \includegraphics[width=\textwidth]{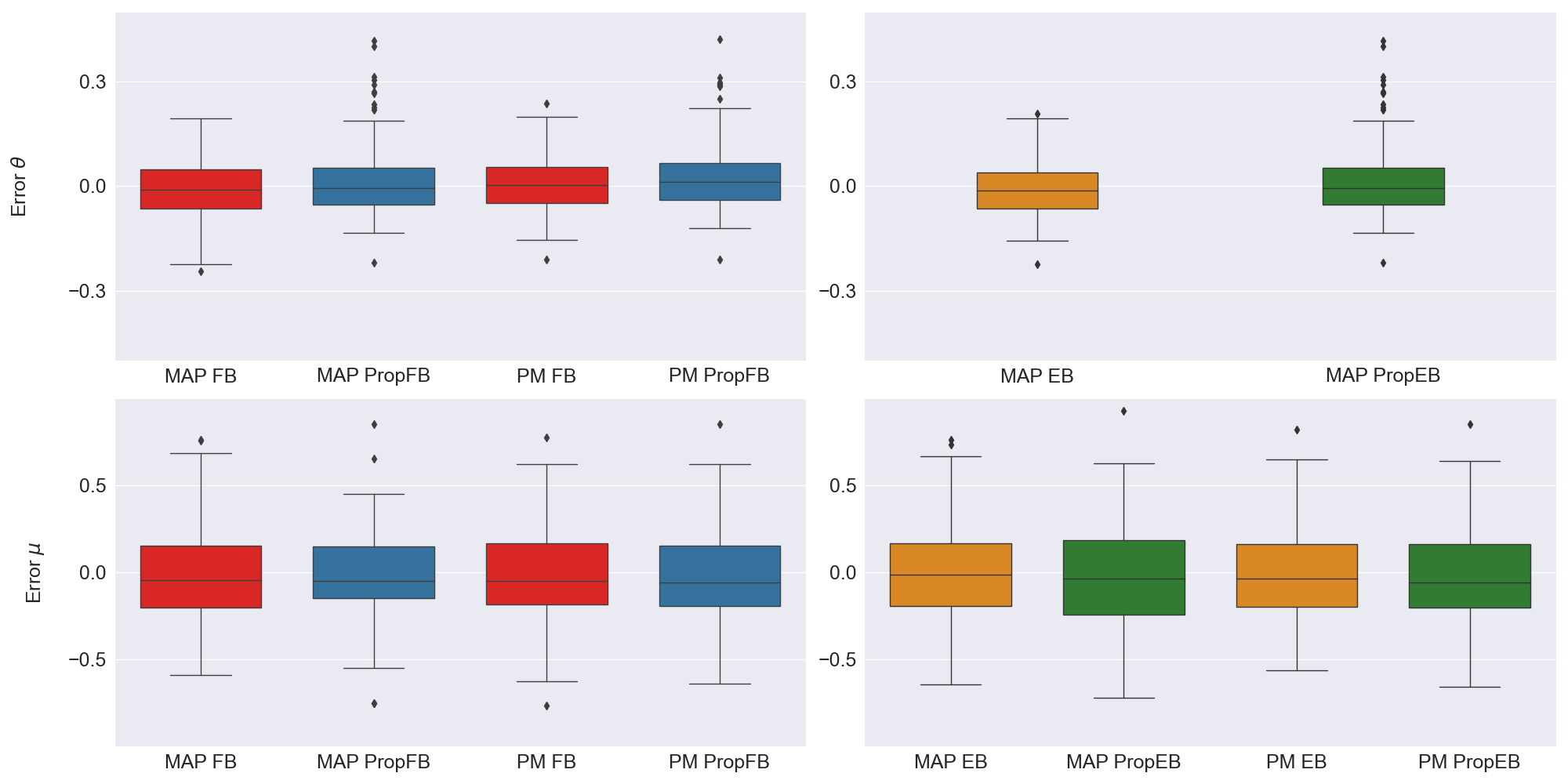}
    \caption{Box-plots for the error in the estimation of the hyper-parameter $\theta$ (first row), the error in the estimation of the parameter $\mu$ (second row).
      There are shown respectively in red, blue, yellow and green the error committed considering the MAP and PM estimates for the Fully Bayesian, Proposed Fully Bayesian, Empirical Bayes and Proposed Empirical Bayes approaches.}
       \label{Fig:ErrorParameters}
\end{figure*}

We analyze the performances in terms of selection of the parameter and hyper-parameter considering the Posterior Mean (PM) and the maximum a posteriori (MAP) estimators, and compute the Euclidean distance between the true and the estimated value of the (hyper-)parameter.

In Figure \ref{Fig:ErrorParameters} we report the boxplots obtained for the (hyper-)parameter estimation error over the 100 datasets:

\begin{itemize}
    \item in the first row (left panel) we show the error in the estimation of the hyper-parameter considering as estimate the MAP and the PM respectively for the FB approach and the PropFB one. 
    
    \item in the first row (right panel) we show the error in the estimation of the hyper-parameter considering as estimate the MAP for the EB and PropEB approach. Due to the structure of the proposed algorithm, the MAP estimate for the PropEB is the same obtained with the PropFB. 
    
    \item in the second row (left panel) we show the error in the estimation of the parameter committed respectively by the FB and PropFB approaches considering as estimate the MAP and the PM.
    
    \item in the second row (right panel) we show the error in the estimation of the parameter committed respectively by the EB and PropEB approaches.
\end{itemize}

We notice that the proposed approach features similar performances as the alternative approaches, either the EB and the FB, in terms of estimation error, while keeping a substantially lower computational cost (Fig. \ref{Fig:CpuTimeToy}). In the case of the FB approach, the proposed method also features a larger ESS (Fig. \ref{Fig:EssToy}).

\begin{figure}
    \centering    
    \includegraphics[width=\columnwidth]{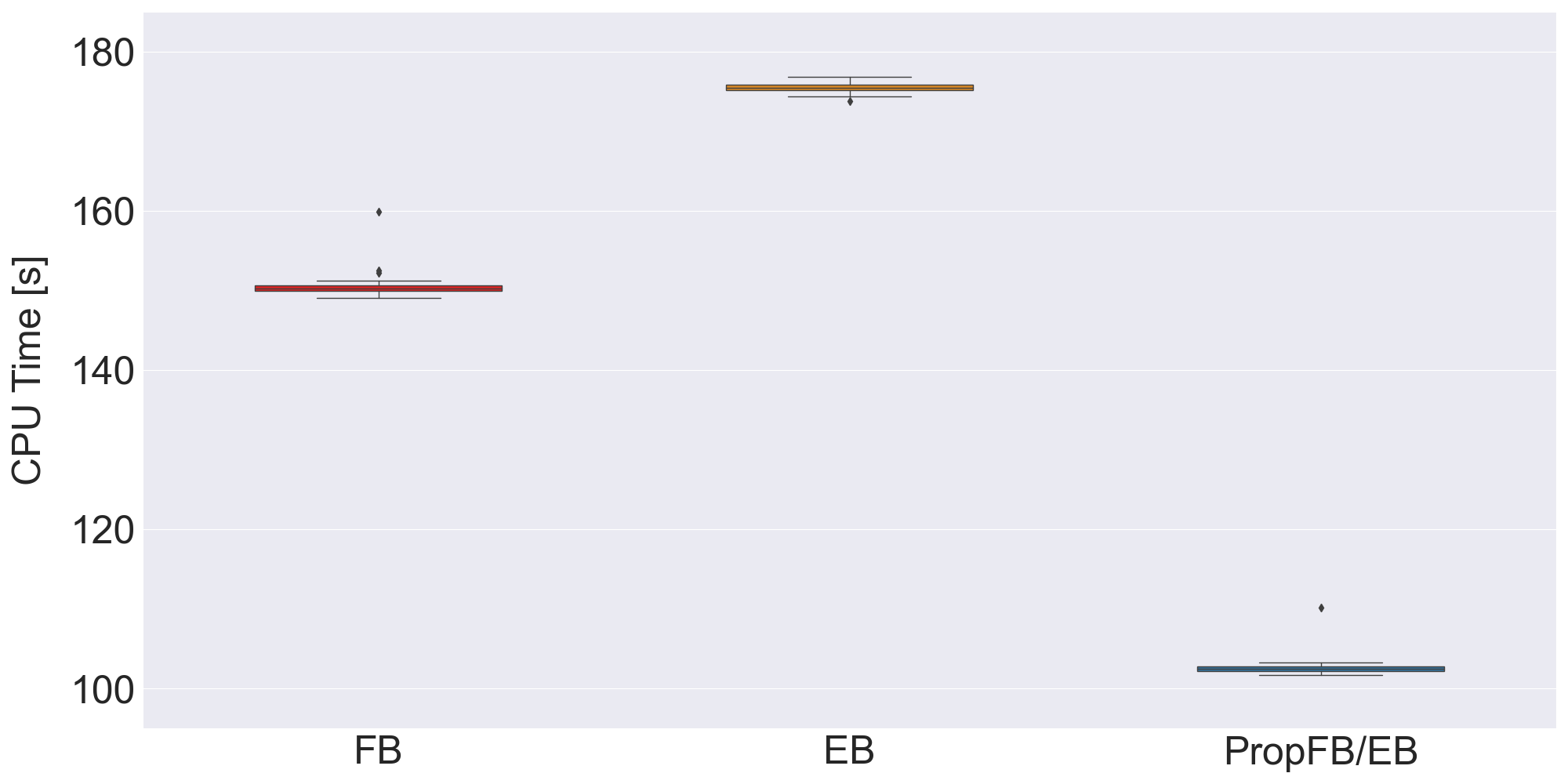}
    \caption{Computational time for the utilized methods. Performances are referred to a MacBook Pro (13-inch, M1, 2020) with 8 GB of memory.}
    \label{Fig:CpuTimeToy}
\end{figure}

\begin{figure}
    \centering    
    \includegraphics[width=\columnwidth]{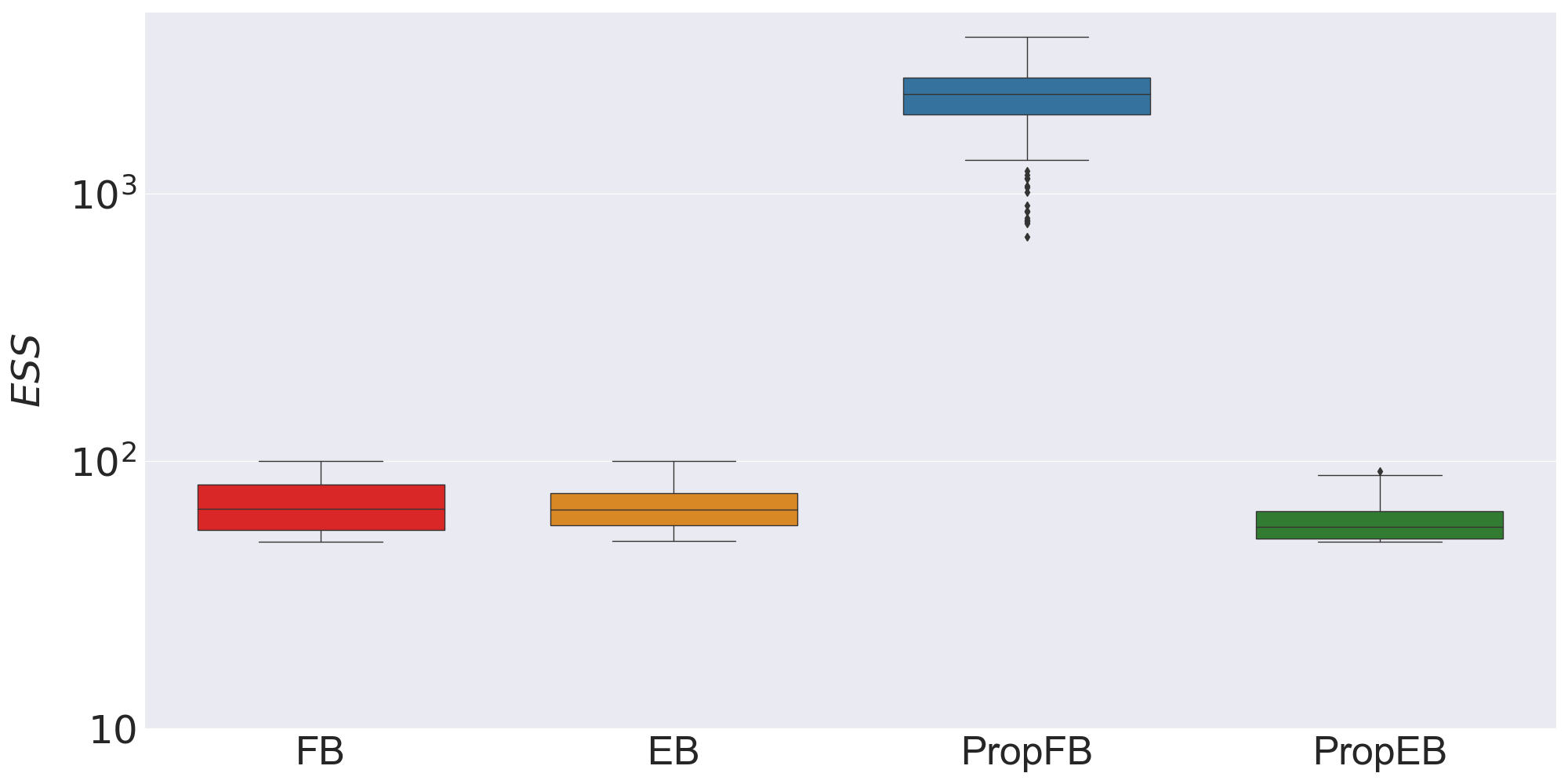}
    \caption{Effective Sample size respectively for the Fully Bayesian, Empirical Bayes, Proposed Fully Bayesian and Proposed Empirical Bayes approaches.}
    \label{Fig:EssToy}
\end{figure}

\begin{figure*}
    \centering    
    \includegraphics[width=\textwidth]{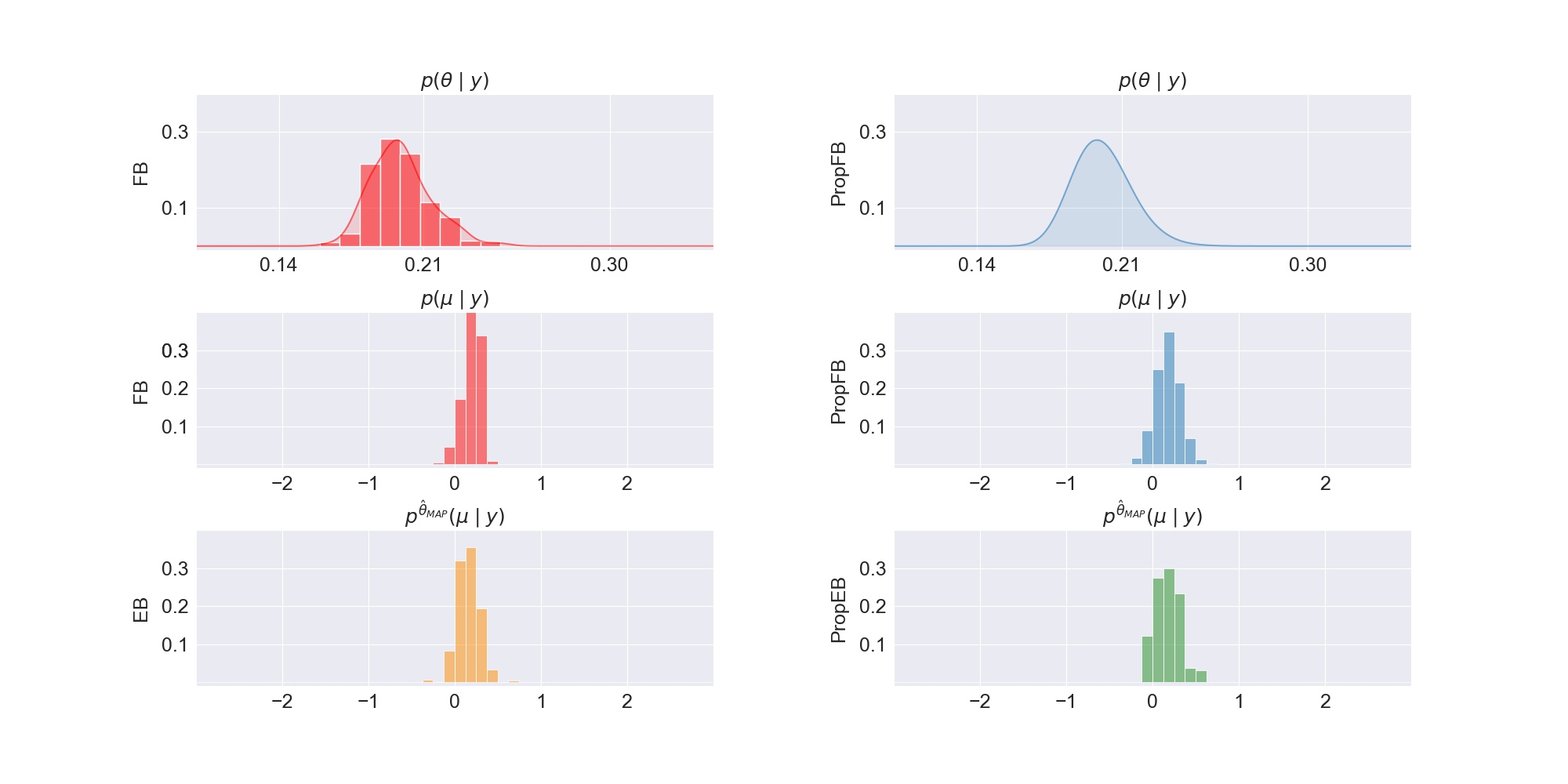}
    \caption{Illustrative example of the posterior for the hyper-parameter (first row), the marginal of the joint posterior for the parameter (second row) and the conditional posterior for the parameter (third row).}
    \label{Fig:OneExampleToy}
\end{figure*}

\subsubsection{Sample result}
For illustrative purposes, in this Section we show results from one specific dataset taken from the 100 simulations used in the previous Section. 

In Figure \ref{Fig:OneExampleToy} we show the output obtained by the proposed method and by the two alternative approaches, specifically by showing:
\begin{itemize}
    \item in the first row, the approximated posterior distribution for the hyper-parameter;
    \item in the second row, the approximated posterior distribution for the parameter obtained in a FB approach;
    \item in the third row, the approximated posterior distribution for the parameter in an EB approach.
\end{itemize}
As far as the approximation of the marginal posterior of the hyper-parameter is concerned, both approximations peak around the correct value, i.e. $\theta_{\textrm{true}}=0.20$. Regarding the approximations for the posterior of the parameter, we observe that all the approximated distributions peak at a positive value but contain the true value (zero) well within their support. .

\section{Application to source imaging in Magneto/ Electro-EncephaloGraphy}
\label{Sec:ApplicationMEG}
In this Section we present the results obtained with the application of the Rao-Blackwellized SMC samplers with the proposed method described in Section \ref{Sec:Selection/AveragingRaoBlackwell} for the resolution of the M/E-EG inverse problem \cite{sommariva2014sequential} introduced as motivating example in Section \ref{Sec:MotivatingExample}.

\begin{figure}
    \centering    
    \includegraphics[width=\columnwidth]{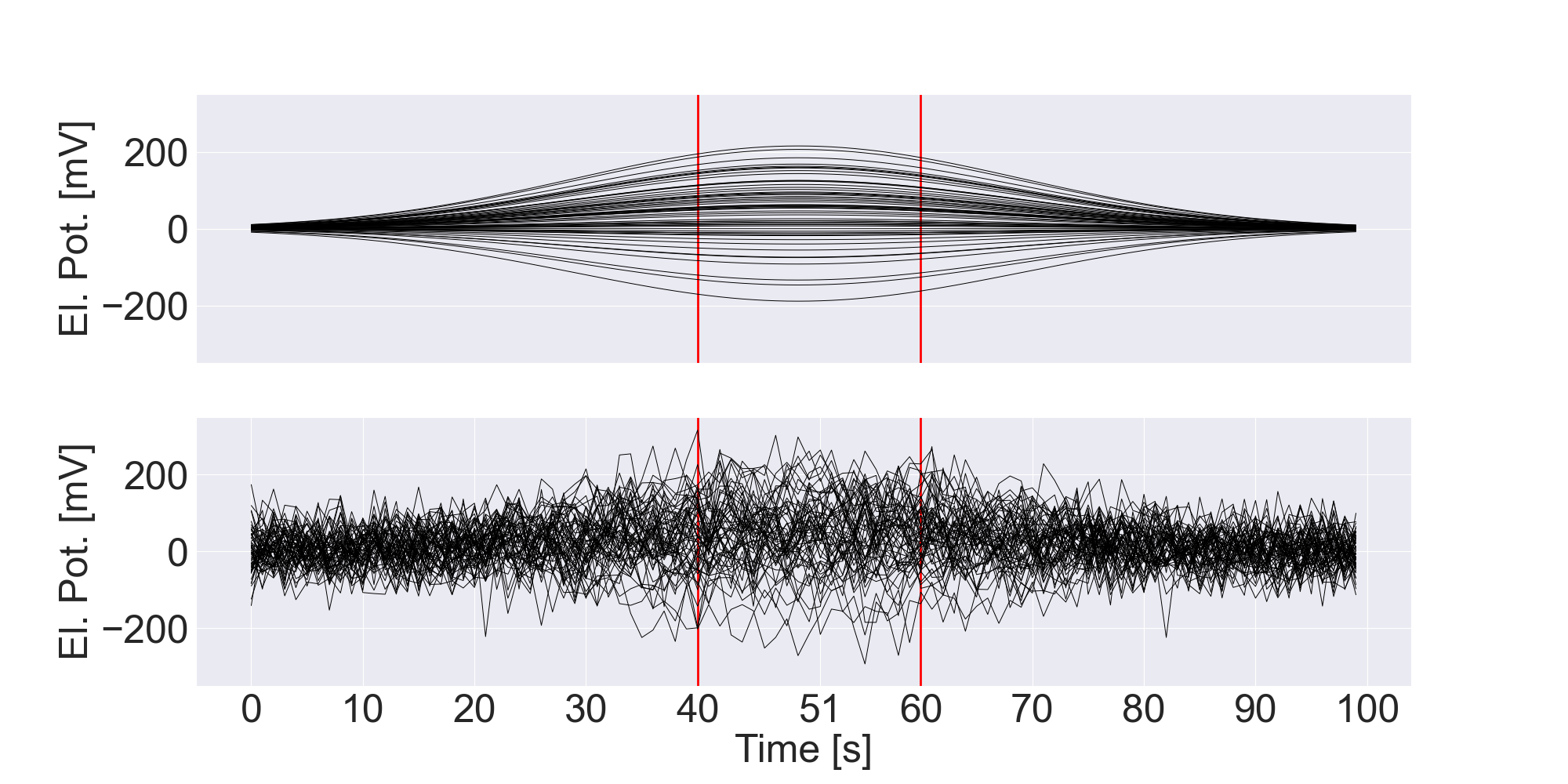}
    \caption{Example of the simulated data used for the analysis as recordings performed by $59$ EEG channels each second from $0$ to $100$. In the first row of the plot we show the synthetic data without noise while in the second row we show the same data with the addition of additive Gaussian noise. The region between the red vertical lines is that which is observed.}
    \label{Fig:DataMEG}
\end{figure}

\begin{figure*}
    \centering    
    \includegraphics[width=\textwidth]{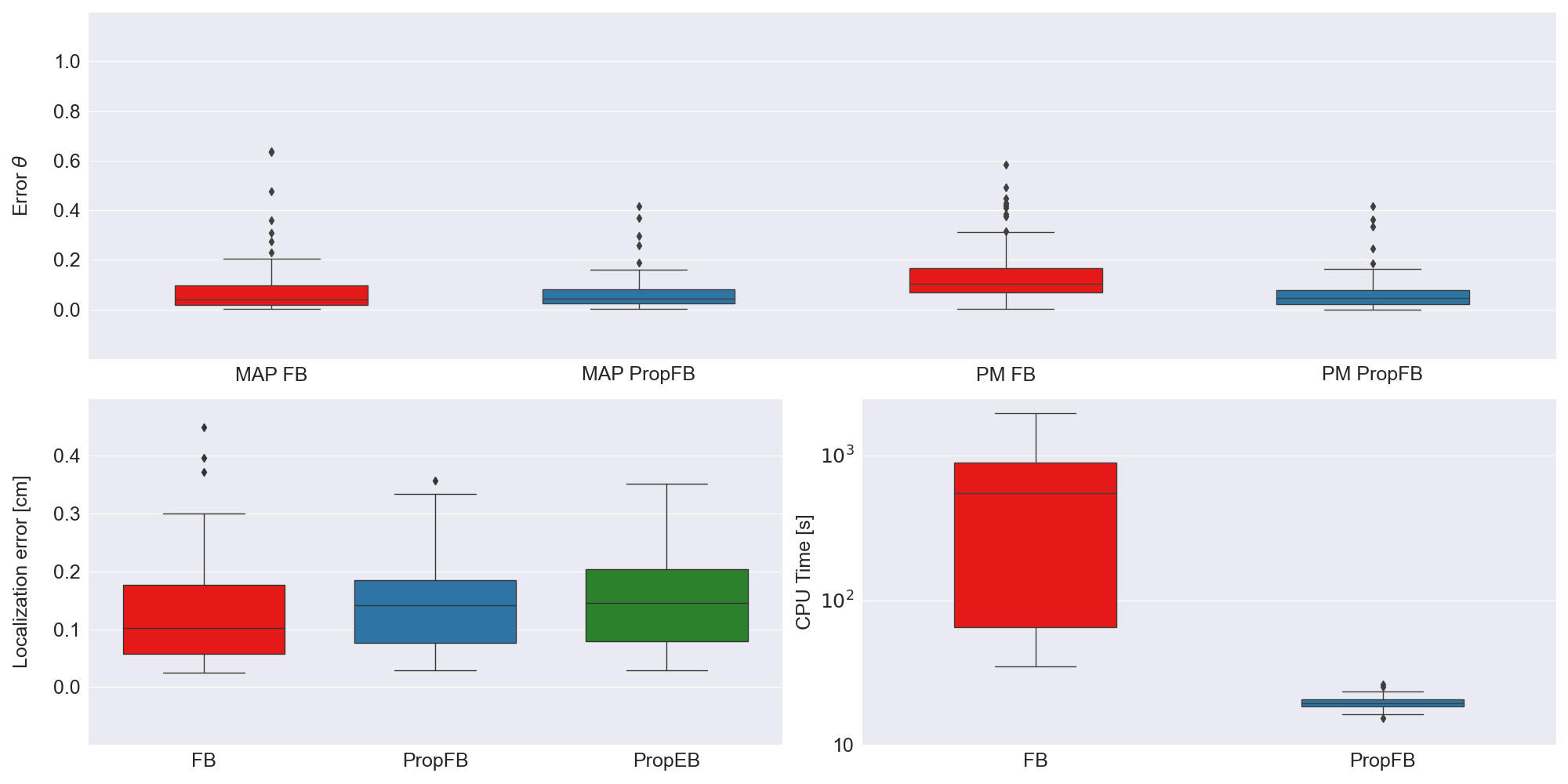}
    \caption{Relative error for the approximation of the hyper-parameter $\theta$ (first row), the boxplots for the OSPA metric in centimeters (second row, left panel) and the computational time in seconds (second row, right panel)  considering the Fully Bayesian (red), the Proposed Fully Bayesian (blue) and the Proposed Empirical Bayes (green) approaches. The computational time is referred to a MacBook Pro (13-inch, M1, 2020) with 8 GB of memory.}
    \label{Fig:errors_MEG}

\end{figure*}

\subsection{Data Generation}
Data $\mathbf{y}=(y(1),\dots,y(T))$ are generated with the following configuration:
\begin{itemize}
    \item brain discretization $\Omega$ with $8193$ voxels;
    \item number of M/E-EG channels: $59$;
    \item number of dipoles: $d=2$;
    \item dipole position $r_i$: randomly drawn, with uniform distribution among the voxels, with the constraint that the distance between the two dipoles is larger than 3cm; the constraint was set in order to allow for identifiability of the two dipoles;
    \item dipole moment $q_i$: orientation chosen among the three orthogonal directions, as the one that maximizes signal strength; unit dipole strength;
    \item noise standard deviation: $\theta_{\textrm{true}}\sim \mathcal{U}[1, 100]$.
\end{itemize}
With these settings, we generate 50 independent realizations of the dataset in order to test the proposed algorithm; Figure \ref{Fig:DataMEG} shows one example of the obtained data.

\subsection{Prior and likelihood}
We assume that all parameters are a priori independent, being $x = (d, \lambda, r_{1:d})$, the prior density is therefore
\begin{equation}
    p(x) = p(d)p(\lambda) \prod_{i=1}^{d} p(r_i),
    \label{Eq:prod_prior_MEG}
\end{equation}

where we specify:
\begin{itemize}
    \item $p(d)\sim \textrm{Poisson}(1)$;
    \item $p\left(\log(\lambda)\right)\sim \mathcal{U}\left(\left[-8,\;-5\right]\right)$;
    \item $p(r_i)\sim\mathcal{U}\left(\Omega\right)$.
\end{itemize}

We assume that noise is not correlated in time, corresponding to conditional independence between data recorded at different time points; the likelihood thus factorizes
    \begin{equation}
        p^{\theta}(\mathbf{y}\mid x) = \prod_{t=1}^T p^{\theta}(y(t) \mid d, \lambda, r_{1:d}).
        \label{Eq:prod_likelihood_MEG}
    \end{equation}

\subsection{Algorithm settings}

Each SMC sampler was applied with the following settings:
\begin{itemize}
    \item analysis window corresponding to the interval $[40,60]$, as shown in figure \ref{Fig:DataMEG}, i.e. analysis windows centered in the peak of the signal;
    \item number of particles set to $100$ as a compromise between performances and quality of the approximation;
    \item $\theta^\star=\min\{\theta_{\textrm{true}}\}/2$; this allows the true value $\theta_{\textrm{true}}$ to be within the range of values explored by the Proposed method during SMC sampler iterations; the order of magnitude of noise is typically known in this kind of data, therefore it would not be difficult to apply a similar reasoning to experimental data;
     
    \item number of iterations set to $100$, with the sequence of exponents growing exponentially in order to guarantee a smooth transit between intermediate distributions;
    \item resampling step performed by means of systematic resampling \cite{douc2005comparison} whenever the effective sample size is lower than half of the number of particles;
    \item MCMC kernels as described in \cite{sommariva2014sequential}.
\end{itemize}

\begin{figure*}
    \centering    
    \includegraphics[width=\textwidth]{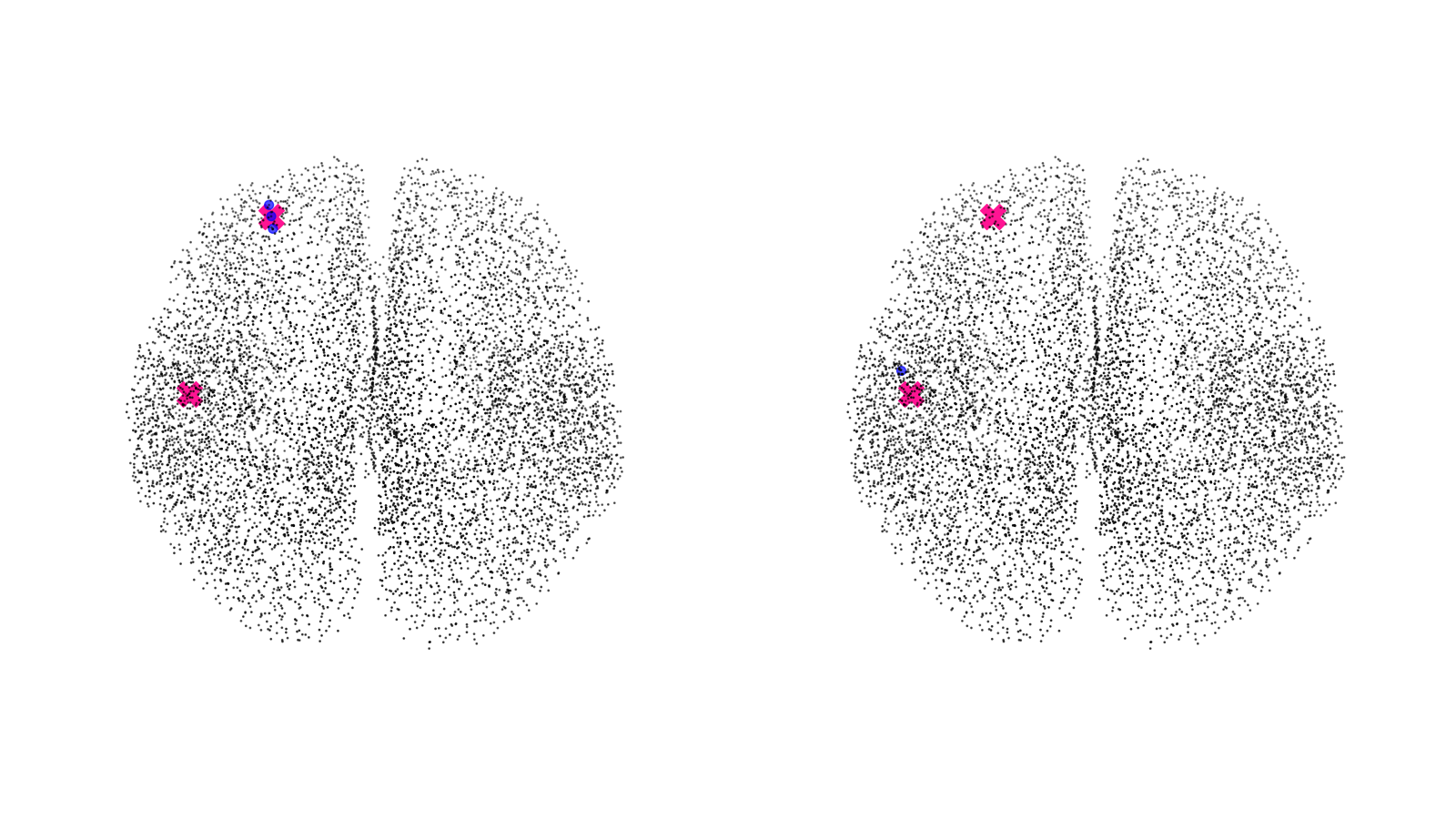}
    \caption{Posterior probability maps for source localization obtained with the Fully Bayesian (left panel) and the Proposed Fully Bayesian (right panel) approaches. Results are visualized on a discretized brain as black dots, the blue dots represent the high probability regions while the purple cross are the estimated dipoles.}
    \label{Fig:ExampleSolutionMEG}
\end{figure*}

\subsection{Performance metrics}
We consider the performances in terms of selection of the hyper-parameter and in terms of localization of current dipoles.

The estimates considered for the hyper-parameter are the MAP and the PM of the marginal posterior $p(\theta \mid \mathbf{y})$, while the estimates for the number and the localization are defined as: 
\begin{itemize}
    \item estimator for number of dipoles: $\hat{d} = \textrm{arg\,max}_{d \in \mathbb{N}}\left(p(d \mid \mathbf{y})\right)$
    \item estimator for dipole location: we construct $\hat{d}$ clusters and than obtain $\hat{r}_{i}$, for $i=1,\dots,\hat{d}$, as the peak of the marginal posterior $p(r \mid \mathbf{y}, \hat{d})$ in the $i$-th cluster.
\end{itemize}
We note that the location estimates are a little-nonstandard in the statistics literature, but this strategy is widespread in the mutiple-object tracking literature (see, e.g., \cite{sorrentino2013dynamic}) as a natural solution to the label-switching problem in this context. 

As the number of dipoles is estimated from the data, the true and estimated number of dipoles might differ; for this reason, in order to evaluate the localization error we consider the  Optimal Sub-Pattern Assignment (OSPA) metric \cite{ristic2011metric}, defined as follows:

\begin{equation}
    OSPA(\hat{r}_{1:\hat{d}}, r_{1:d}) = \min_{\phi}\sum_{i=1}^{\min\{\hat{d}, d\}} \|\hat{r}_{i} - r_{\phi(i)}\|
    \label{Eq:OSPA}
\end{equation}
where the minimum is taken over all possible permutations, $\phi$, of $\{1,\ldots,d\}$.

\subsection{Results}
In Figure \ref{Fig:errors_MEG} we report the boxplots obtained for the hyper-parameter estimation and the localization error over the 50 datasets:

\begin{itemize}
    \item in the first row we show the error in the estimation of the hyper-parameter considering as estimate the MAP and the PM respectively for the FB approach and the PropFB one;
    \item in the left panel of the second row we show the OSPA metric respectively for the FB, PropFB and PropEB approaches;
    
    \item in the right panel of the second row we report the CPU time.
\end{itemize}

Our results indicate that the proposed approach performs slightly better than the alternative in terms of hyper-parameter estimation, while localization error as measured by the OSPA metric is not significantly different. The computational cost of the proposed approach is considerably lower than the one of the alternative approach; the difference is much more evident than in the case of the toy example. This difference can be explained by the combined effect of the variable dimension model, i.e. the SMC sampler exploring spaces with different number of sources, and the sampling of the hyper-parameter: when the sampled hyper-parameter gets small, the SMC sampler tends to prefer configurations with larger number of sources whose likelihood calculation is more expensive.
 
\begin{figure}
    \centering
    \includegraphics[width=\columnwidth]{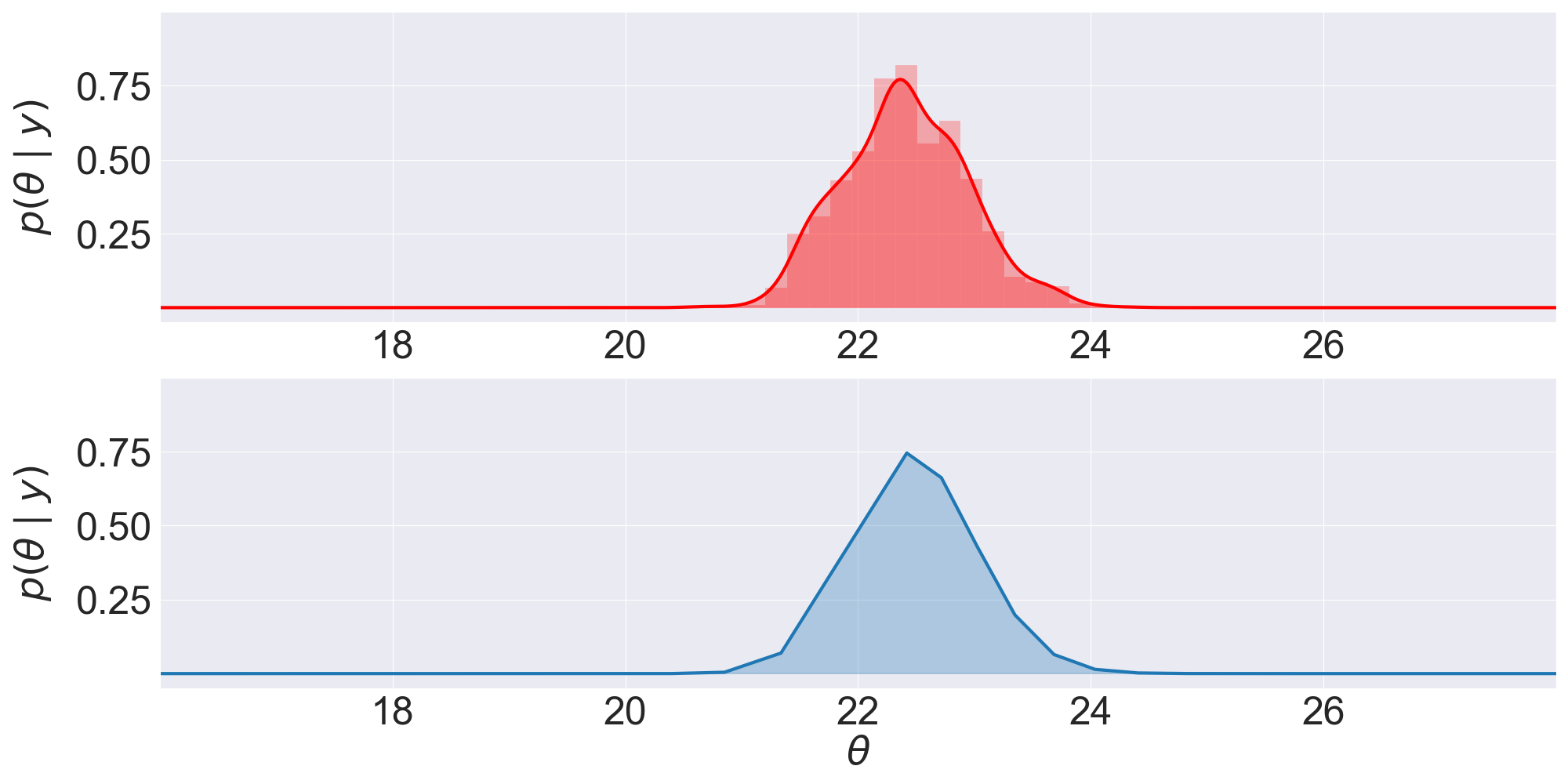}
    \caption{The figure shows the posterior for the parameter $\theta$ approximated with the Fully Bayesian approach (red) and with the Proposed Fully Bayesian (blue).}
    \label{fig:PosteriorThetaMEG}
\end{figure}

\subsubsection{Sample result}
For illustrative purposes, in this Section we show results from one specific dataset taken from the 50 simulations used in the previous Section.

In Figure \ref{Fig:ExampleSolutionMEG} we show the posterior distribution for the source location $p(r \mid \mathbf{y}, \hat{d})$ approximated respectively by the FB (left panel) and the PropFB (right panel) approaches. Both posterior are reciprocally similar and both methods estimate two sources, one in the left and one in the right hemisphere.

In Figure \ref{fig:PosteriorThetaMEG} we show the approximated posterior distributions for the hyper-parameter provided by the two algorithms. Again we can observe that the two approximations are similar to each other and peaked around the correct value. $\theta_{\textrm{true}}=22$.

\section{Conclusions}
\label{Sec:Conclusions}
We presented a method that allows to perform at the same time and with a limited computational cost Fully Bayesian and Empirical Bayes approaches.

Experiments show that the method performs slightly better than the natural alternatives, but with important differences.
The proposed approach is more versatile in several ways: it allows to compute maximum likelihood/a posteriori estimates of the hyper-parameter; it allows to recycle the SMC samples for a different hyper-prior; it allows hyper-parameter selection via marginal maximum likelihood, and to provide estimates of the unknown for a specific value of the hyper-parameter. In addition, when it comes to averaging across different values of the hyper-parameter, it provides substantially more Monte Carlo samples, potentially allowing better approximations of the posterior and resulting in better estimates of the unknowns.

Importantly, all these advantages are obtained essentially for free, i.e. at no additional computational cost; in addition, the proposed approach exploits samples at all iterations, thus simultaneously overcoming one of the known limitations of SMC samplers, i.e. the fact that intermediate samples are usually discarded.

Finally, although this article is dedicated to exploiting the particular structure present in a class of problems with univariate hyper-parameters in a way which yields both standard and empirical Bayesian estimates simultaneously with little overhead, it also suggests a path to efficiently performing empirical Bayesian estimation in a broader class of models. By estimating the gradient of the marginal likelihood with respect to the hyperparameter using the current particle set would in principle allow the adaptive specification of a sequence of hyper-parameter values (and hence posterior distributions) which converges towards that which maximises the marginal likelihood. This is beyond the scope of this manuscript but provides an interesting avenue for future exploration.

\bmhead{Acknowledgments}
AMJ acknowledges support from the EPSRC (grant numbers  EP/R034710/1 and EP/T004134/1) and the Lloyd's Register Foundation Programme on Data-Centric Engineering at the Alan Turing Institute.
AV and AS kindly acknowledge Gruppo Nazionale per il Calcolo Scientifico (GNCS) for partial support. 

\begin{appendices}
\section{}
\label{secA1}

\begin{proposition} 
\label{Prop:NEF-Exponential}
Let $p^{\theta} \in NEF$ with sufficient statistic $T$ and canonical parameter $\theta$ s.t. $p^{\theta}(x) = \exp(\theta T(x)-A_{\theta})$
and $\alpha\neq0$, then: 
\small
\begin{equation*}
    [p^{\theta}(x)]^{\alpha} = \exp(A_{\alpha\theta}-\alpha A_{\theta}) p^{\alpha\theta}(x)
\end{equation*}
\end{proposition}

\begin{proof}
$$[p^{\theta}(x)]^{\alpha} = \exp(\alpha\theta T(x)-\alpha A_{\theta}) = \exp(A_{\alpha\theta}-\alpha A_{\theta}) p^{\alpha\theta}(x)$$
\end{proof}
\normalsize

\begin{corollary} 
\label{Cor:exponential_gaussian}
Let $p(x\mid\sigma) = \mathcal{N}(x ;\; \mu,\sigma^2\Gamma)$ be an m-dimensional Gaussian density, then for any $\alpha\neq0$:
\small
\begin{equation*}
    p(x\mid\sigma)^{\alpha} = \sqrt{\left((2\pi)^m \det(\sigma^2\Gamma)\right)^{1-\alpha}\alpha^{-m}} p\left(x\;\mid\;\frac{\sigma}{\sqrt{\alpha}}\right).
\end{equation*}
\end{corollary}
\normalsize
\begin{proof} The proof follows by the consideration that the family of the considered densities is a NEF with:
\small
\begin{itemize}
    \item $\theta = \frac{1}{\sigma^2}$
    \item $T(x) = -\frac{1}{2}(x-\mu)^t\Gamma^{-1}(x-\mu)$
    \item $\exp(A_{\theta}) = \left((2\pi)^m \det\left(\frac{1}{\theta}\Gamma\right)\right)^{\frac{1}{2}}$
\end{itemize}
\normalsize
Therefore the previous result guarantees the thesis because of the normalizing constant is given by
\small
\begin{equation*}
    \begin{split}
        \exp&(A_{\alpha\theta}-\alpha A_{\theta})=\\ & 
        \sqrt{\left((2\pi)^m \det\left(\frac{1}{\alpha\theta}\Gamma\right)\right) \left((2\pi)^m \det\left(\frac{1}{\theta}\Gamma\right)\right)^{-\alpha}}= \\ & 
        \sqrt{\left((2\pi)^m \det\left(\frac{1}{\theta}\Gamma\right)\right)^{1-\alpha}\alpha^{-m}}= \\ & 
        \sqrt{\left((2\pi)^m \det(\sigma^2\Gamma)\right)^{1-\alpha}\alpha^{-m}}
    \end{split}
\end{equation*}
\end{proof}

\normalsize
\begin{proposition} 
\label{Prop:marginal_like}
Let $p(x_1\mid\lambda) \sim \mathcal{N}(\eta,\Gamma_{\lambda})$ be an $m$-dimensional Gaussian density and consider a $k$-dimensional Gaussian density $p^{\theta}(y\mid x_1,x_2,\lambda) \sim \mathcal{N}(\mu(x_2)x_1,\Sigma_{\theta})$, assuming that $x_1$ is independent of $x_2$:
\small
\begin{equation*}
\begin{split}
    &p(x_1\mid\lambda)p^{\theta}(y \mid x_1,x_2,\lambda) = \\& 
    \mathcal{N}\left(
    \begin{bmatrix}
      x_1 \\
      y
    \end{bmatrix} ;
    \;
    \begin{bmatrix}
\eta \\
\eta\mu(x_2)
\end{bmatrix}
,
\begin{bmatrix}
\Gamma_\lambda && \Gamma_\lambda \mu(x_2)^t \\
\mu(x_2)\Gamma_\lambda && \Sigma_{\theta}+\mu(x_2)\Gamma_\lambda \mu(x_2)^t
\end{bmatrix}\right),
\end{split}
\label{Proof:part_a}
\end{equation*}
\begin{equation*} 
p^{\theta}(y \mid x_2,\lambda) \sim \mathcal{N}\left(y;\; \eta\mu(x_2), \Sigma_{\theta}+\mu(x_2)\Gamma_\lambda \mu(x_2)^t\right).\label{Proof:part_b} 
\end{equation*}
\end{proposition}
\normalsize
\begin{proof}
Without loss of generality we assume that $\eta = 0$, therefore the product of the Gaussian densities turns out to be
\small
\begin{equation*}
        \begin{split}
            p&(x_1\mid\lambda)p^{\theta}(y \mid x_1,x_2,\lambda) \\ 
            & \propto \exp\left((y-\mu(x_2)x_1)^t\Sigma_{\theta}^{-1}(y-\mu(x_2)x_1)+x_1^t\Gamma_\lambda^{-1}x_1\right)\\
            & =\exp\bigl(y^t\Sigma_{\theta}^{-1}y - x_1^t\mu(x_2)^t\Sigma_{\theta}^{-1}y + x_1^t\mu(x_2)^t\Sigma_{\theta}^{-1}\mu(x_2)x_1 \\
            & - y^t\Sigma_{\theta}^{-1}\mu(x_2)x_1 + x_1^t\Gamma_\lambda^{-1}\Sigma_{\theta}\bigl)\\
            & =\exp\bigl(y^t\Sigma_{\theta}^{-1}y - x_1^t\mu(x_2)^t\Sigma_{\theta}^{-1}y+x_1^t(\mu(x_2)^t\Sigma_{\theta}^{-1}\mu(x_2) \\
            & + \Gamma_\lambda^{-1})x_1 - y^t\Sigma_{\theta}^{-1} \mu(x_2)x_1\bigl)\\
            & =\exp\biggl(\begin{bmatrix}
                        x_1\\
                        y
                        \end{bmatrix}^t
                         \\ &
                        \begin{bmatrix}
                        \Gamma_\lambda^{-1}+\mu(x_2)^t\Sigma_{\theta}^{-1}\mu(x_2) && -\mu(x_2)^t\Sigma_{\theta}^{-1} \\
                         -\Sigma_{\theta}^{-1}\mu(x_2) && \Sigma_{\theta}^{-1}
                        \end{bmatrix}
                        \begin{bmatrix}
                        x_1\\
                        y
                        \end{bmatrix}\biggl)
         \end{split}
\end{equation*}
\normalsize
with the normalizing constant

\begin{equation}
\small
\begin{split}
\biggl((2\pi)^{m}\det(\Sigma_{\theta})& (2\pi)^{k}\det(\Gamma_\lambda)\biggl)^{-\frac{1}{2}}=\\ &\left((2\pi)^{m+k}\det(\Sigma_{\theta})\det(\Gamma_\lambda)\right)^{-\frac{1}{2}}.
\end{split}
\end{equation}

\normalsize
If we consider the multivariate normal density
\small
\begin{equation}
        \mathcal{N}\left(\begin{bmatrix}
x_1 \\
y
\end{bmatrix}
;
\;
\begin{bmatrix}
0 \\
0
\end{bmatrix}
,
\begin{bmatrix}
\Gamma_\lambda && \Gamma_\lambda \mu(x_2)^t \\
\mu(x_2)\Gamma_\lambda && \Sigma_{\theta}+\mu(x_2)\Gamma_\lambda \mu(x_2)^t
\end{bmatrix}\right)
\end{equation}
\normalsize
then the inverse of the covariance matrix, thanks to a classical result of block-matrix inversion, turns out to be
\small
\begin{equation*}
    \begin{split}
    &\begin{bmatrix}
    \Gamma_\lambda && \Gamma_\lambda \mu(x_2)^t \\
    \mu(x_2)\Gamma_\lambda && \Sigma_{\theta}+\mu(x_2)\Gamma_\lambda \mu(x_2)^t
    \end{bmatrix}^{-1} \\ 
    & = \begin{bmatrix} \Gamma_\lambda && \Gamma_\lambda \mu(x_2)^t \\ \mu(x_2)\Gamma_\lambda && \Sigma_{\theta}+\mu(x_2)\Gamma_\lambda \mu(x_2)^t \end{bmatrix}^{-1} \\
    &= \begin{bmatrix} \Gamma_\lambda^{-1} (I+\Gamma_\lambda \mu(x_2)^t\Sigma_{\theta}^{-1}\mu(x_2)\Gamma_\lambda\Gamma_\lambda^{-1})&& -\Gamma_\lambda^{-1}\Gamma_\lambda \mu(x_2)^t\Sigma_{\theta}^{-1} \\ -\Sigma_{\theta}^{-1}\mu(x_2)\Gamma_\lambda\Gamma_\lambda^{-1} && \Sigma_{\theta}^{-1} \end{bmatrix}\\
    &=\begin{bmatrix} \Gamma_\lambda^{-1} (I+\Gamma_\lambda \mu(x_2)^t\Sigma_{\theta}^{-1}\mu(x_2)) && -\mu(x_2)^t\Sigma_{\theta}^{-1} \\ -\Sigma_{\theta}^{-1}\mu(x_2) && \Sigma_{\theta}^{-1} \end{bmatrix} \\
    & = \begin{bmatrix} \Gamma_\lambda^{-1} +\mu(x_2)^t\Sigma_{\theta}^{-1}\mu(x_2) && -\mu(x_2)^t\Sigma_{\theta}^{-1} \\ -\Sigma_{\theta}^{-1}\mu(x_2) && \Sigma_{\theta}^{-1} \end{bmatrix},
    \end{split}
\end{equation*}
\normalsize
Where the normalizing constant is
\small
$$\left(\left(2\pi\right)^{\frac{m+k}{2}}\det\left(\begin{bmatrix}
\Gamma_\lambda  && \Gamma_\lambda \mu(x_2)^t \\
\mu(x_2)\Gamma_\lambda  && \Sigma_{\theta}+\mu(x_2)\Gamma_\lambda \mu(x_2)^t
\end{bmatrix}\right)\right)^{-\frac{1}{2}}.$$
\normalsize
where the determinant of the covariance matrix is equal to
\small
\begin{equation*}
\begin{split}
&\det\left(\Sigma_{\theta}+\mu(x_2)\Gamma_\lambda \mu(x_2)^t\right) \\& \det\left(\Gamma_\lambda -\Gamma_\lambda \mu(x_2)^t (\Sigma_{\theta}+\mu(x_2)\Gamma_\lambda \mu(x_2)^t)^{-1}\mu(x_2)\Gamma_\lambda\right)\\
& =\det\left(\Sigma_{\theta} (I+\Sigma_{\theta}^{-1}\mu(x_2)\Gamma_\lambda \mu(x_2)^t)\right)  \det(\Gamma_\lambda)\\& \det\left(I-\mu(x_2)^t (\Sigma_{\theta}+\mu(x_2)\Gamma_\lambda \mu(x_2)^t)^{-1}\mu(x_2)\Gamma_\lambda\right)\\
& =\det(\Sigma_{\theta}) \det\left(I+\Sigma_{\theta}^{-1}\mu(x_2)\Gamma_\lambda \mu(x_2)^t\right) \\& \det(\Gamma_\lambda )\det\left((I+\mu(x_2)^t \Sigma_{\theta}^{-1}\mu(x_2)\Gamma_\lambda )^{-1}\right)\\
& =\det(\Sigma_{\theta}) \det\left(I+\Sigma_{\theta}^{-1}\mu(x_2)\Gamma_\lambda \mu(x_2)^t\right)\det(\Gamma_\lambda) \\& \det\left(I+\Sigma_{\theta}^{-1}\mu(x_2)\Gamma_\lambda \mu(x_2)^t\right)^{-1}\\
& =\det(\Sigma_{\theta})\det(\Gamma_\lambda).
    \end{split}
\end{equation*}
\normalsize
Therefore, from a well known result on Gaussian densities, we obtain the thesis
\small
\begin{equation*}
    \begin{split}
    p^{\theta}(y \mid x_2,\lambda)&=\int p^{\theta}(y,x_1\mid x_2,\lambda)dx_1 \\
    & = \int p^{\theta}(y\mid x_1,x_2,\lambda)p(x_1\mid x_2,\lambda)dx_1 \\
    & =\mathcal{N}\left(y;0, \Sigma_{\theta}+\mu(x_2)\Gamma_\lambda \mu(x_2)^t\right).
    \end{split}
\end{equation*}
\end{proof}

\end{appendices}

\bibliographystyle{plain}
\bibliography{bibliography}


\end{document}